\documentclass[conference]{IEEEtran}
\IEEEoverridecommandlockouts
\usepackage{amsmath,amsfonts}
\usepackage{amssymb,amsmath,amsthm}

\usepackage[letterpaper, top=1.78cm, bottom=2.57cm, left=1.7cm, right=1.7cm]{geometry}

\newtheoremstyle{custom}
  {0.5em} 
  {0.5em} 
  {\itshape} 
  {0.5em} 
  {\itshape\space} 
  {:} 
  {0.5em} 
  {\thmname{#1}\thmnumber{\textit{\hspace{0.5em}#2}}\thmnote{#3}} 

\theoremstyle{custom}

\newtheorem{Lemma}{Lemma}

\newtheorem{Proposition}{Proposition}
\newtheorem{Result}{Result}
\newtheorem{Property}{Property}

\makeatletter
\renewenvironment{proof}[1][\proofname]{\par
  \pushQED{\qed}%
  \normalfont \topsep6\p@\@plus6\p@\relax
  \trivlist
  \itemindent1.5em 
  \item[\hskip\labelsep
        \itshape
        #1\@addpunct{:}]\ignorespaces
}{%
  \popQED\endtrivlist\@endpefalse
}
\makeatother

\usepackage{bbm}
\usepackage{algorithm}
\usepackage{array}
\usepackage[caption=false,font=normalsize,labelfont=sf,textfont=sf]{subfig}
\usepackage{textcomp}
\usepackage{stfloats}
\usepackage{url}
\usepackage{verbatim}
\usepackage{graphicx}
\usepackage{cite}
\usepackage{color}
\def\BibTeX{{\rm B\kern-.05em{\sc i\kern-.025em b}\kern-.08em
    T\kern-.1667em\lower.7ex\hbox{E}\kern-.125emX}}

\newcommand{\LSDC}{\mathtt{LS}}
\newcommand{\MSDC}{\mathtt{MS}}
\newcommand{\Nf}{N_\mathrm{f}}
\newcommand{\Ns}{N_\mathrm{s}}
\newcommand{\NC}{N_{\mathrm{C}}}

\begin{document}


\title{Availability of Aerial Heterogeneous Networks for Reliable Emergency Communications}

\author{

\thanks{This work is supported in part by Natural Science Foundation of China (Grant No. 62121001), in part by Key Research and Development Program of Shannxi (Grant No. 2024CY2-GJHX-82), and in part by China Scholarship Council (Student ID: 202506960051). The work of M. Matthaiou was supported by the European Research Council (ERC) under the European Union’s Horizon 2020 research and innovation programme (grant agreement No. 101001331).}

\IEEEauthorblockN{
\fontsize{0.36cm}{1cm}\selectfont
Teng~Wu$^\dag$$^\ddagger$, Jiandong~Li$^\dag$, Junyu~Liu$^\dag$, Min~Sheng$^\dag$, Mohammadali Mohammadi$^\ddagger$,  Hien Quoc Ngo$^\ddagger$, and Michail Matthaiou$^\ddagger$}
\IEEEauthorblockA{
$^\dag$State Key Laboratory of ISN, Institute of Information Science, Xidian University, Xi’an, Shaanxi, 710071, China\\
$^\ddagger$Centre for Wireless Innovation (CWI), Queen's University Belfast, Belfast, BT3 9DT, United Kingdom\\
Email: t.wu@stu.xidian.edu.cn}}


\maketitle

\begin{abstract}
We investigate network availability (NA) in aerial heterogeneous networks (AHetNets) for effective emergency rescue, where diverse delay-constrained communication services must be provided to user equipments (UEs) with varying mobility. The heterogeneity in delay constraints and UE mobility introduces resource allocation conflicts and imbalances, which undermine communication reliability and challenge NA. Although unified resource allocation (URA) can mitigate these issues, it remains unclear whether NA can be sustained under such diverse conditions.
To address this, we derive expressions for the lower bound (LB) on NA in AHetNets under URA. Our analysis reveals that extended heterogeneity significantly degrades the LB due to resource limitations—even when the heterogeneity stems from additional services under less stringent delay constraints (LSDC) or from UEs with lower mobility. To overcome this degradation, we formulate and solve a joint optimization problem for the number of UEs sharing time-frequency resources ($K$) and pilot length ($\xi$), aiming to enhance the LB by improving spatial, frequency, and temporal resource efficiency. Simulation results validate our analysis and demonstrate that jointly optimizing $K$ and $\xi$ enables AHetNets to achieve the target NA under greater heterogeneity, outperforming existing resource allocation policies.
\end{abstract}


\vspace{-1em}
\section{Introduction}
Providing emergency communications to counteract the detrimental effects of physical disasters is a key use case of sixth-generation (6G) globalization \cite{Dao2021Tuts}. Benefiting from flexibility and adaptability, unmanned aerial vehicle (UAV) communications are particularly suitable for emergencies \cite{Yao2022MNET}. On one hand, UAVs equipped with thermal imagers, infrared scanners, and cameras as aerial UEs (AUEs) can collect vital information for search and rescue. On the other hand, aerial networks with UAV-mounted flying access points (FAPs) can be swiftly deployed for providing emergency communications to AUEs and ground UEs (GUEs) of rescuers and survivors. 

Despite their potential, emergency communication services for disaster rescue face formidable challenges. Unreliability arises from dynamic channels due to FAP movement, wireless fading, and service burstiness \cite{Liu2025iotj,LiuTCOM2025}. Additionally, AUEs and GUEs exhibit distinct mobility and require communications under diverse delay constraints—from second-level delays for real-time media to sub-millisecond delays for remote control \cite{Yao2022MNET}. This heterogeneity causes resource allocation conflicts and imbalances \cite{Liu2025iotj}, undermining reliability. As a result, the NA of AHetNets is hindered from supporting effective emergency rescue, where NA is the probability that the quality-of-service (QoS) requirements, in terms of delay and reliability, for each UE's service are satisfied \cite{LiuTCOM2025}.

Recent works suggest that coordinated multi-point (CoMP) can address unreliability from dynamic channels, wireless fading, and service burstiness, even when AUEs and GUEs coexist \cite{LiuTCOM2025, Mohammadi2024Proc, Elwekeil:TCOM:2023}. Effective resource allocation under CoMP for diverse delay constraints \cite{Liu2025TVT} and UE mobility \cite{Zhang2025TVT} can enhance transmission rates and reliability. However, existing works largely overlook how delay constraints and UE mobility jointly affect NA. Most studies focus on a single delay constraint or UE mobility type \cite{LiuTCOM2025}, and although URA appears suitable for heterogeneous networks \cite{Abbas2025IoTJ}, its viability under diverse delay and mobility remains unclear. Consequently, NA in AHetNets is not yet well understood, and current resource allocation strategies cannot be directly applied to achieve the target NA.

Motivated by the above, this work investigates the NA of FAP CoMP-enabled AHetNets for reliable emergency communication (REC) services under a URA scheme. The main contributions of this paper are as follows:
\vspace{-0.1em}
\begin{itemize}
    \item We derive a LB on NA for performance analysis and show that it is significantly degraded under extended heterogeneity due to reduced resources in URA, even when heterogeneity stems from additional REC services under LSDC or from UEs with lower mobility.
    \item Our analysis reveals that jointly optimizing the number of UEs sharing time-frequency resources ($K$) and pilot length ($\xi$) can effectively improve spatial, frequency, and temporal resource efficiency, thereby enhancing the LB on NA. Building on this insight, we formulate and solve a joint optimization problem for $K$ and $\xi$ to enhance NA.
    \item  Our simulation results confirm that the target NA is achievable even under greater heterogeneity by jointly optimizing $K$ and $\xi$, outperforming existing policies.
\end{itemize}  

\vspace{-0.1em}
\textit{Notations}: ${\mathbb{E}_X}[  \cdot  ]$ denotes the expectation with respect to (w.r.t.) the random variable (RV) $X$; ${\mathbb{P}_X}(  \cdot  )$ denotes the probability w.r.t. the RV $X$; $\lfloor  \cdot  \rfloor $ and $\lceil  \cdot  \rceil $ are the floor and ceil operators, respectively; $\mathbbm{1}[ Y ]$ denotes the indicator function of event $Y$, where $\mathbbm{1}[ Y ]$ = 1 if event $Y$ is true, and $\mathbbm{1}[ Y ]$ = 0 otherwise; $z \sim {\cal C}{\cal N}( {0,{\sigma ^2}} )$ denotes a complex Gaussian RV $z$ with zero mean and variance $\sigma ^2$; ${f_{{Q^{ - 1}}}}(  \cdot  )$ is the inverse Q-function, while the Q-function is ${f_Q}( x ) = \frac{1}{{\sqrt {2\pi } }}\int_x^\infty  {\exp ( { - \frac{{{y^2}}}{2}} )dy} $. Finally, ${J_0}(  \cdot  )$ is the zeroth-order Bessel function of the first kind.

\section{System Model}\label{Sec_System_Model}
\vspace{-0.1cm}
\subsection{Network Model}
\vspace{-0.1cm}
We consider an uplink FAP CoMP-enabled AHetNet, which consists of $L$ fixed-wing UAV-mounted single-antenna FAPs and $M$ single-antenna UEs with diverse mobility. In general, $M \gg L$. All FAPs are connected to a central processor (CP) via wireless fronthaul links for centralized signal processing, while the CP can be deployed on the ground or in the sky \cite{LiuTCOM2025}. We consider that the disaster occurs in a circular area with a radius $W_{\mathrm{D}}$, which is denoted as the disaster area ${\cal A}$. The UEs include AUEs and GUEs, where GUEs are randomly distributed within ${\cal A}$, while AUEs are randomly distributed over the altitude range ${\cal H} \!=\! \left[ {{h_{\min }},{h_{\max }}} \right]$ above ${\cal A}$. Moreover, ${h_{\min }}$ and ${h_{\max }}$ are the minimum and maximum altitudes where AUEs are located, respectively. All FAPs follow periodic circular flight trajectories. The flight altitude, radius, and velocity of each FAP are $h_{\mathrm{F}}$, $W_{\mathrm{F}}$, and $v_{\mathrm{F}}$, respectively. Thus, the flight period of each FAP is $T \!=\! {{2\pi {W_{\mathrm{F}}}} \mathord{\left/  {\vphantom {{2\pi {W_{\mathrm{F}}}} {{v_{\mathrm{F}}}}}} \right.  \kern-\nulldelimiterspace} {{v_{\mathrm{F}}}}}$. We equally divide $T$ into $I$ time slots with a duration ${T_{\mathrm{S}}}$, while ${T_{\mathrm{S}}}$ is generally chosen as a value less than the coherence interval ${T_{\mathrm{C}}}$, i.e., $I \!=\! {T \mathord{\left/
{\vphantom {T {{T_{\mathrm{S}}}}}} \right.
\kern-\nulldelimiterspace} {{T_{\mathrm{S}}}}}$, $t \!=\! i{T_{\mathrm{S}}}$, $\forall t \!\in\! \left[ {0,T} \right]$, $\forall i \!\in\! \left[ {0,I} \right]$, and ${T_{\mathrm{S}}} \!\le\! {T_{\mathrm{C}}}$. The UE mobility in AHetNets is characterized by the relative velocity ${v_{\mathrm{r}}}$ between UEs and FAPs. We consider that all FAPs have the same ${v_{\mathrm{r}}}$ w.r.t. the same UE, while ${v_{\mathrm{r}}}$ may vary across different UEs due to diverse mobility. Moreover, we consider that each transmission distance from UEs to FAPs $d  \! = \! \sqrt {d_{\mathrm{H}}^2 + {h^2}} $ is constant during one time slot and varies between different time slots \cite{LiuTCOM2025}, where ${d_{\mathrm{H}}}$ and $h$ are the horizontal and vertical distances from UEs to FAPs, respectively. The farthest transmission distance from UEs to FAPs is ${\tilde d}  \!\! = \!\! \sqrt  {\!{{{\tilde d}_{\mathrm{H}}}^2} \!+\! {{{\tilde h}}^2}}$, where ${{\tilde d}_{\mathrm{H}}}  \! = \! 2{W_{\mathrm{D}}} + 2{W_{\mathrm{F}}}$ and ${\tilde h}  \! = \! \max \left\{ {h_{\mathrm{F}},\left| {h_{\mathrm{F}} \!-\! {h_{\min }}} \right|,\left| {h_{\mathrm{F}} \!-\! {h_{\max }}} \right|} \right\}$ are the farthest horizontal and vertical distances from UEs to FAPs, respectively.

\subsection{Traffic Model and Scheduling}
\vspace{-0.1cm}
In general, RECs are delay-sensitive services with burstiness, where delay constraints primarily include LSDC and more stringent delay constraints (MSDC) categories \cite{Yao2022MNET, Dong2021TWC, Liu2025iotj}. The service burstiness can be characterized by the variance of the random arrival rates \cite{Liu2025iotj}. Thus, the arrival processes of REC services are modeled as stochastic arrival processes with average arrival rate $\theta^{\mathtt{\varsigma}}$ and arrival rate variance $\sigma _\mathtt{\varsigma} ^2$, $\mathtt{\varsigma}  \in   \left\{ \LSDC, \MSDC \right\}$, where $\LSDC$ and $\MSDC$ denote REC services under LSDC and MSDC, respectively. To manage unreliability due to bursty services, we consider deploying a first-come-first-serve server in the system for uplink traffic scheduling, which is a common and effective method \cite{LiuTCOM2025, Liu2025iotj}. As a result, there will be a queue at the buffer in the UE. 
\par

\vspace{-0.1cm}
\subsection{Channel Model}\label{SubSec_channel_model}
\vspace{-0.1cm}
We consider both quasi-static and time-varying channels due to diverse delay constraints and UE mobility \cite{LiuTCOM2025,Elhoushy2021TWC}.

\subsubsection{Quasi-Static Channel}
The channels from UEs to FAPs are independent identically distributed (i.i.d.) during different coherence intervals and can be modeled as \cite{LiuTCOM2025}
\begin{equation}\label{eq_channel_quasi_static}
\setlength{\abovedisplayskip}{0pt}
\setlength{\belowdisplayskip}{0pt}
    {g} = \sqrt {\beta} {\psi},
\end{equation}
where ${\psi}$ denotes the small-scale and shadow fading that is modeled as $\kappa$-$\mu$ shadowed fading; $\beta$ denotes the path loss that is assumed to be known a priori and modeled as $\beta \!\left( d \right) = \mathcal{A}d^{ - 2 }$; $\mathcal{A}\left( {{\mathrm{dB}}} \right) =  - 20{\log _{10}}\left( {{{4\pi } \mathord{\left/
 {\vphantom {{4\pi } \lambda}} \right.
 \kern-\nulldelimiterspace} \lambda}} \right)$ denotes the path loss at 1 m with wavelength $\lambda = {{{v_{\mathrm{L}}}} \mathord{\left/
 {\vphantom {{{v_{\mathrm{L}}}} {{f_{\mathrm{c}}}}}} \right.
 \kern-\nulldelimiterspace} {{f_{\mathrm{c}}}}}$; ${f_{\mathrm{c}}}$ is the carrier frequency, while ${v_{\mathrm{L}}}$ is the speed of light.

\subsubsection{Time-Varying Channel}
The channels from UEs to FAPs between adjacent coherence intervals are correlated. We apply a first-order auto-regressive (AR) model to portray the channel correlation \cite{Elhoushy2021TWC}. In that, the channels can be modeled as
\begin{equation}\label{eq_channel_time_varying}
g = e_{\mathrm{c}} g_{-1} + \sqrt {1 - {e_{\mathrm{c}} ^2}} e_{\mathrm{unc}},
\end{equation}
where $e_{\mathrm{unc}} \sim {\cal C}{\cal N}\left( {0,\beta } \right)$ denotes the innovation component during one coherence interval that is uncorrelated with the channel during the previous coherence interval ${g_{ - 1}}$; $e_{\mathrm{c}} $ is the temporal correlation coefficient between adjacent coherence intervals that can be calculated by using Jakes fading model. Specifically, $e_{\mathrm{c}}  = {J_0}\left( {2\pi {f_{\mathrm{D}}}{t_{\mathrm{s}}}} \right)$, where ${{t_{\mathrm{s}}}}$ is the sampling time and ${f_{\mathrm{D}}} = {{{v_{\mathrm{r}}}{f_{\mathrm{c}}}} \mathord{\left/
 {\vphantom {{{v_{\mathrm{r}}}{f_{\mathrm{c}}}} {{v_{\mathrm{L}}}}}} \right.
 \kern-\nulldelimiterspace} {{v_{\mathrm{L}}}}}$ is the corresponding Doppler frequency shift of ${v_{\mathrm{r}}}$ between UEs and FAPs. Finally, the coherence interval is ${T_{\mathrm{C}}} \left({v_{\mathrm{r}}}\right) \triangleq {{9{v_{\mathrm{L}}}} \mathord{\left/
 {\vphantom {{9{v_{\mathrm{L}}}} {\left( {16\pi {f_{\mathrm{c}}}{v_{\mathrm{r}}}} \right)}}} \right.
 \kern-\nulldelimiterspace} {\left( {16\pi {f_{\mathrm{c}}}{v_{\mathrm{r}}}} \right)}}$ \cite{LiuTCOM2025}. Note that \eqref{eq_channel_time_varying} is applicable for the channels during non-initial coherence intervals of each service, while the channels during the initial coherence interval of each service are modeled as \eqref{eq_channel_quasi_static}.

\section{Preliminary Analysis}
\vspace{-0.1cm}
\subsection{QoS Requirements} \label{SubSec_QoS}
\vspace{-0.1cm}
For each UE's uplink REC service, the QoS requirements, in terms of delay and reliability, are characterized by the total delay bound in the uplink $D_{\max }^{ \mathtt{\varsigma} }$ and the requirement of the uplink overall packet loss (uOPL) probability $\varepsilon _{\max }^{ \mathtt{\varsigma} }$ \cite{LiuTCOM2025}. Thus, the QoS requirements $\left( {D_{\max }^{\mathtt{\varsigma}} ,\varepsilon _{\max }^{\mathtt{\varsigma}} } \right)$ for each UE's uplink REC service can be satisfied under the following constraints:
\begin{equation}\label{eq_QoS_req_define}
    {D^{\mathtt{\varsigma} }} \le D_{\max }^{\mathtt{\varsigma} },{\bar \varepsilon ^{  \mathtt{\varsigma} }} \le \varepsilon _{\max }^{  \mathtt{\varsigma} },
\end{equation}
where ${D^{ \mathtt{\varsigma} }}$ is the total delay in the uplink; ${{\bar \varepsilon }^{\mathtt{\varsigma}} } = \mathbb{E}_{\psi}\left[ {\varepsilon _{\scriptscriptstyle \psi } ^{\mathtt{\varsigma}} } \right]$ is the uOPL probability averaged over the small-scale and shadow fading; $\varepsilon _{\scriptscriptstyle \psi } ^{\mathtt{\varsigma}}$ is the uOPL probability under instantaneous fading; $\mathtt{\varsigma}  \in \left\{ \LSDC,\MSDC \right\}$. Services under both MSDC and LSDC exhibit diverse delay constraints, where $D_{\max }^{{\MSDC}}$ is sub-millisecond level and $D_{\max }^{{\LSDC}}$ is millisecond level or higher \cite{Dong2021TWC}.  \par

According to the 3rd Generation Partnership Project (3GPP) specifications, the total delay in the uplink includes, but is not limited to, the transmission delay, queueing delay, propagation delay, as well as coding and processing delay \cite{3gpp_22_261_TS}. This paper focuses on managing spatial, frequency, and temporal resources, which pertain to the transmission delay and queueing delay \cite{LiuTCOM2025, Liu2025iotj}. Therefore, this paper considers that the total uplink delay ${D^{ \mathtt{\varsigma} }}$ consists of the uplink transmission delay ${D^{\mathrm{u}}}$ and queueing delay ${D^{{\mathrm{q,}}\mathtt{\varsigma} }}$, i.e., ${D^{ \mathtt{\varsigma} }} = {D^{\mathrm{u}}} + {D^{{\mathrm{q,}}\mathtt{\varsigma} }}$. To satisfy delay constraints, we employ a short frame structure \cite{LiuTCOM2025, Liu2025iotj}. The frame duration ${T_{\mathrm{f}}}$ equals the transmission time interval (TTI), which is the minimal time granularity of the network. Each frame can include data transmission part with duration $T_{\mathrm{f}}^{\left( {\mathrm{d}} \right)}$ and control signaling part with duration $T_{\mathrm{f}}^{\left( {\mathrm{c}} \right)}$. In this paper, we design the control signaling as orthogonal pilots to estimate channel state information (CSI) \cite{LiuTCOM2025, Liu2025iotj}. Each uplink transmission takes one frame. The queueing delay is bounded as $D_{\max }^{{\mathrm{q,}}\mathtt{\varsigma} } \buildrel \Delta \over = D_{\max }^{ \mathtt{\varsigma} } - {D^{\mathrm{u}}}$.

According to \cite{LiuTCOM2025, Liu2025iotj}, we consider that the transmitted packet will be discarded if any decoding error occurs at the receiver. Moreover, if the queueing delay exceeds its bound, the transmitted packet will be discarded. As a result, the decoding error probability and queueing delay violation probability should be considered in the uOPL probability ${\varepsilon_{\scriptscriptstyle \psi } ^{  \mathtt{\varsigma} }}$ to effectively reflect reliability. \par

\subsection{Network Availability}
The NA is defined as the probability that the QoS requirements, in terms of delay and reliability, for each UE's service are satisfied \cite{LiuTCOM2025}. Thus, the NA can be expressed as 
\begin{equation}\label{eq_NA_define}
   {\eta  } \triangleq \frac{1}{M}{\sum\nolimits_{m = 1}^M  {\mathbbm{1}\left[ {{{\left( {{D^{\mathtt{\varsigma} }} \le D_{\max }^{\mathtt{\varsigma} },{\bar \varepsilon ^{\mathtt{\varsigma} }} \le \varepsilon _{\max }^{  \mathtt{\varsigma} }} \right)}_m}} \right]} },
\end{equation}
where ${{{\left( {{D^{\mathtt{\varsigma} }} \!\le\! D_{\max }^{\mathtt{\varsigma} },{\bar \varepsilon ^{\mathtt{\varsigma} }} \!\le\! \varepsilon _{\max }^{  \mathtt{\varsigma} }} \right)}_m}}$ is an event that the delay and reliability requirements for the $m$th UE's service are satisfied. \par

Given a target NA ${\eta _{\max }}$, achieving $\left( {{\eta } \ge {\eta _{\max }},{\eta _{\max }}  \to 1} \right)$ indicates that each UE's REC service is ensured. However, the NA definition in \eqref{eq_NA_define} fails to intuitively reveal the impact of heterogeneity that stems from diverse delay constraints and UE mobility. Moreover, the NA evaluation in aerial networks must account for both temporal and spatial dimensions due to FAP movement \cite{LiuTCOM2025}. Consequently, the NA analysis in AHetNets becomes highly complex. Inspired by \cite{Imran2024Proc}, we consider a worst-case scenario to explore the LB on NA, where all UEs are activated and each transmission distance from UEs to FAPs is the farthest, i.e., $d = {\tilde d}$. As a result, the NA evaluation for services with identical QoS requirements of UEs with same mobility is simplified because the spatial and temporal dimensions can be omitted, which facilitates the NA analysis in AHetNets. This simplification is articulated in \textit{Lemma} \ref{Lemma_NA_max_evaluation}.
\begin{Lemma}\label{Lemma_NA_max_evaluation}
    The NA for services with identical QoS requirements $\left( {D_{\max }^{\mathtt{\varsigma}} ,\varepsilon _{\max }^{\mathtt{\varsigma}} } \right)$ of UEs with same mobility can be lower bounded by ${{\mathbb{P}_\psi }\left( {{D^{\mathtt{\varsigma}} } \le D_{\max }^{\mathtt{\varsigma}} , \varepsilon _{\scriptscriptstyle \psi } ^{\mathtt{\varsigma}} \le \varepsilon _{\max }^{\mathtt{\varsigma}} } \right) }$ of any UE under the condition that ${d = {\tilde d}}$.
\end{Lemma}
\begin{proof}
    The LB on NA is due to the fact that NA will be degraded as $d$ increases \cite{LiuTCOM2025}.
\end{proof}

\subsection{URA-based CoMP Packet Delivery Mechanism}\label{Sec_Packet_Delivery_Mechanism}
In the URA scheme, $M$ UEs are partitioned into $K$ proximity-based clusters, each containing at most $N = \left\lceil {{M \mathord{\left/
{\vphantom {M K}} \right.
\kern-\nulldelimiterspace} K}} \right\rceil $ UEs. This strategy is independent of the delay constraints, UE categories, and mobility. The total bandwidth ${B^{{\mathrm{tot}}}}$ is evenly divided into $N$ orthogonal subchannels, where each UE occupies one subchannel. The per-subchannel bandwidth obeys ${B_0} \le B = {{N}_0}{B_0} \le {B_{\mathrm{C}}}$ \cite{Liu2025iotj,LiuTCOM2025}, where ${B_0}$ is the orthogonal subcarrier spacing, $1 \le {N_0} = \left\lfloor {{{{B^{{\mathrm{tot}}}}} \mathord{\left/
{\vphantom {{{B^{{\mathrm{tot}}}}} {\left( {{B_0}N} \right)}}} \right.
\kern-\nulldelimiterspace} {\left( {{B_0}N} \right)}}} \right\rfloor  \le \left\lfloor {{{{B_{\mathrm{C}}}} \mathord{\left/
{\vphantom {{{B_{\mathrm{C}}}} {{B_0}}}} \right.
\kern-\nulldelimiterspace} {{B_0}}}} \right\rfloor $ is the number of orthogonal subcarriers in each subchannel, while ${B_{\mathrm{C}}}$ is the coherence bandwidth. Thus, the number of orthogonal subchannels is constrained by
\begin{equation}\label{eq_constraint_N}
     {N_{\mathrm{I} }} = \left\lfloor {\frac{{{B^{{\mathrm{tot}}}}}}{{{B_{\mathrm{C}}}}}} \right\rfloor  \le N \le {N_{\mathrm{II} }} = \left\lfloor {\frac{{{B^{{\mathrm{tot}}}}}}{{{B_0}}}} \right\rfloor .
\end{equation}
The UEs within a cluster are interference-free by availing of orthogonal subchannels. Across clusters, UEs sharing time-frequency resources form a UE group whose mutual interference is managed by FAP CoMP. There are $N$ UE groups. Within a group, UEs are referred to as scheduled UEs. Due to $N = \left\lceil {{M \mathord{\left/
{\vphantom {M K}} \right.
\kern-\nulldelimiterspace} K}} \right\rceil $ and the constraint in \eqref{eq_constraint_N}, there are at most $K$ scheduled UEs in each group and $K$ is constrained by
\begin{equation}\label{eq_constraint_K}
     {K_{\mathrm{I} }} \! = \! \max \left\{ {1,\left\lceil {\frac{M}{{{N_{\mathrm{II} }}}}} \right\rceil } \right\} \le K\mathop  \le \limits^{\left( a \right)} \min \left\{ {L,\left\lceil {\frac{M}{{{N_{\mathrm{I} }}}}} \right\rceil } \right\} \! = \! {K_{\mathrm{II} }},
\end{equation}
where step $(a)$ in \eqref{eq_constraint_K} is because $K \le L$ and reflects the condition for exploiting spatial diversity and multiplexing gains through FAP CoMP \cite{LiuTCOM2025}. We denote the $k$th scheduled UE in the $n$th group as UE$_{n,k}$, $\forall n \in \left[ {1,N} \right]$ and $\forall k \in \left[ {1,K} \right]$. Under constraint \eqref{eq_constraint_K}, the bandwidth allocated to each UE is $B \left( K \right) = \left\lfloor {{{{B^{{\mathrm{tot}}}}} \mathord{\left/
 {\vphantom {{{B^{{\mathrm{tot}}}}} {\left( {{B_0}\left\lceil {{M \mathord{\left/
 {\vphantom {M K}} \right.
 \kern-\nulldelimiterspace} K}} \right\rceil } \right)}}} \right.
 \kern-\nulldelimiterspace} {\left( {{B_0}\left\lceil {{M \mathord{\left/
 {\vphantom {M K}} \right.
 \kern-\nulldelimiterspace} K}} \right\rceil } \right)}}} \right\rfloor {B_0}$.  \par

Based on the aforementioned unified bandwidth allocation, UEs transmit signals that are received by all FAPs. Then, all FAPs forward the received signals to the CP for centralized signal processing, which includes CSI estimation based on the minimum mean-square error estimation method with orthogonal pilots and data detection based on zero-forcing (ZF) detection method under imperfect CSI \cite{LiuTCOM2025}. Centralized signal processing is implemented at every time slot with duration ${T_{\mathrm{S}}} \!=\! {T_0}$, where ${T_0}$ is the greatest common divisor across diverse $D_{\max }^{ \mathtt{\varsigma} }$ and ${T_{\mathrm{C}}} \left({v_{\mathrm{r}}}\right)$. One ${T_{\mathrm{C}}} \left({v_{\mathrm{r}}}\right)$ contains $\Ns \!=\! {{T_{\mathrm{C}}} \left({v_{\mathrm{r}}}\right) \mathord{\left/
{\vphantom {{{T_{\mathrm{C}}} \left({v_{\mathrm{r}}}\right)} {{T_{\mathrm{S}}}}}} \right.
\kern-\nulldelimiterspace} {{T_{\mathrm{S}}}}}$ time slots and ${T_{\mathrm{S}}}$ comprises $\Nf \!=\! {{{T_{\mathrm{S}}}} \mathord{\left/
{\vphantom {{{T_{\mathrm{S}}}} {{T_{\mathrm{f}}}}}} \right.
\kern-\nulldelimiterspace} {{T_{\mathrm{f}}}}}$ frames. To ensure accurate CSI estimation under service burstiness and diverse delay constraints, all UEs transmit pilots in each frame. If no service arrives, the data part of the frame is left empty, while pilot transmission is sustained. Thus, the pilot length $\xi $ for accurate CSI estimation is constrained by \cite{Mohammadi2024Proc}
\begin{equation}\label{eq_constraint_LoP}
     K \le \xi  = B\left( K \right)\Nf T_{\mathrm{f}}^{\left( {\mathrm{c}} \right)} \le {\xi _{\max }} = B\left( K \right)\Nf{T_{\mathrm{f}}}-1,
\end{equation}
where ${\xi _{\max }}$ is the maximum pilot length; $B\left( K \right)\Nf{T_{\mathrm{f}}}$ are the time-frequency resources for each service within ${T_{\mathrm{S}}}$. \par

Based on \cite{LiuTCOM2025}, $g_{i,n,{l},k}$, $\beta _{i,n,l,k}$, and ${\psi _{i,n,l,k}}$ denote the channel, the path loss, and the small-scale and shadowed fading from UE$_{n,k}$ to the ${l}$th FAP during time slot $i$, respectively; $\rho _{\mathrm{p}}( {K,\xi } )  \! = \! {\rho( {K } )  }{\xi  }$ is the pilot sequence power; $\rho ( K ) \! = \! {{\vartheta {p_{\mathrm{t}}}} \mathord{\left/
 {\vphantom {{\vartheta {p_{\mathrm{t}}}} {( {B( K ){{\mathcal N}}} )}}} \right.
 \kern-\nulldelimiterspace} {( {B\left( K \right){{\mathcal N}}} )}}$ and ${p_{\mathrm{t}}}$ are the transmit signal-to-noise ratio (SNR) and transmit power of each UE, respectively; ${{\mathcal N}}$ is the noise power spectral density of the additive white Gaussian noise; $\vartheta \in ( {0,1} )$ is the SNR loss due to wireless fronthaul. 

\begin{Lemma}\label{Lemma_PPSNR}
    Under the URA and the worst-case scenario with $d = {\tilde d}$, the LB on post-processing SNR of UE$_{n,k}$ in FAP CoMP-enabled AHetNets during time slot $i$ can be given by 
\begin{equation}\label{eq_PPSNR_UL_LB}
    \hat \gamma _{i,n,k}^{{\mathrm{ZF}},{\mathrm{LB}}}\left( {K,\xi } \right) = \hat \rho \left( {K,\xi } \right){\tilde \beta}\sum\nolimits_{l = 1}^{L - K + 1} {{{| {{{\hat \psi }_{i,n,l,k}}} |}^2}} ,
\end{equation}
where $\hat \rho \left( {K,\xi } \right) \! = \! {{\rho \left( K \right)} \mathord{\left/
 {\vphantom {{\rho \left( K \right)} {\left( {{K \mathord{\left/
 {\vphantom {K \xi }} \right.
 \kern-\nulldelimiterspace} \xi } + 1} \right)}}} \right.
 \kern-\nulldelimiterspace} {\left( {{K \mathord{\left/
 {\vphantom {K \xi }} \right.
 \kern-\nulldelimiterspace} \xi } + 1} \right)}}$; ${\tilde \beta} \!=\! \beta ( {{\tilde d}} )$ is the path loss under the worst-case scenario; ${{\hat \psi }_{i,n,l,k}} \!=\! {\psi _{i,n,l,k}} - {{\hat w}_{i,n,l,k}}$ is the estimated small-scale and shadowed fading from UE$_{n,k}$ to the ${l}$th FAP during time slot $i$; ${{\hat w}_{i,n,l,k}} \!\sim\! \mathcal{C}\mathcal{N}( {0,{1 \mathord{/
 {\vphantom {1 {( {{\rho _{\mathrm{p}}}( {K,\xi } ){\tilde \beta} + 1} )}}} 
 \kern-\nulldelimiterspace} {( {{\rho _{\mathrm{p}}}( {K,\xi } ){\tilde \beta} + 1} )}}} )$ denotes the channel estimation error of ${\psi _{i,n,l,k}}$ under the worst-case scenario.
\end{Lemma}
\vspace{-0.4cm}
\begin{proof}
    By substituting ${\beta _{i,n,l,k}}$ = $\beta ( {{\tilde d}} )$ into the LB on post-processing SNR in \cite{LiuTCOM2025}, \eqref{eq_PPSNR_UL_LB} can be obtained.
\end{proof}

\section{NA Analysis and Enhancement}\label{Sec_NA_analysis_enhancement}
\vspace{-0.12cm}
\subsection{NA Analysis}
\vspace{-0.12cm}

From Section \ref{SubSec_QoS}, the delay constraints $D_{\max }^{ \mathtt{\varsigma} }$ can be equivalent to the queueing delay constraints due to $D_{\max }^{{\mathrm{q,}}\mathtt{\varsigma} } \triangleq D_{\max }^{ \mathtt{\varsigma} } - {D^{\mathrm{u}}}$. Thus, we consider utilizing effective bandwidth (EB) to analyze the LB on NA, since EB is a principal tool for representing queueing delay constraints \cite{LiuTCOM2025}. Specifically, EB is defined as the minimal constant service rate that is required to satisfy $( {D_{\max }^{\mathrm{q,}\mathtt{\varsigma}},{\varepsilon_{\max } ^{\mathrm{q,}\mathtt{\varsigma}}}})$ for a stochastic arrival process, where $\varepsilon _{\max }^{{\mathrm{q},\mathtt{\varsigma}}}$ is the requirement of queueing delay violation probability. For stochastic arrival process $a^{\mathtt{\varsigma}}\left( t \right)$ (in bits/s), the formal definition of EB $E_{\mathrm{B}}^{\mathtt{\varsigma}}$ (in packets/slot) is given by \cite{She2021Proc} 
\begin{equation}\label{eq_EB_def}
    E_{\mathrm{B}}^{\mathtt{\varsigma}} \! \left( D_{\max }^{{\mathrm{q,}\varsigma} } \right) \! = \! \!\mathop {\lim }\limits_{t \to \infty } \! \frac{{{T_{\mathrm{S}}}}}{{{\varpi ^{\mathtt{\varsigma}}}{\Theta ^{\mathtt{\varsigma}} }t}}\!\ln \mathbb{E}_{a^{\mathtt{\varsigma}}}\!\!\Big[ {\exp \!\Big( {{\Theta ^{\mathtt{\varsigma}} }\!\int_0^t {{a^{\mathtt{\varsigma}} }\left( x \right)dx}\! } \Big)\!} \Big]\!,
\end{equation}
where ${\varpi ^{\mathtt{\varsigma}}}$, $\mathtt{\varsigma}  \in \left\{ \LSDC, \MSDC \right\}$ are the packet sizes (in bits) of the REC services under LSDC and MSDC, respectively; ${\Theta ^{\mathtt{\varsigma}} }$ is the QoS exponent that satisfies 
\begin{equation}\label{eq_EB_def_approx}
    \mathbb{P}_{a^{\mathtt{\varsigma}}}\left( {D ^{{\mathrm{q,}}\mathtt{\varsigma} } \ge D_{\max }^{{\mathrm{q,}\varsigma} }} \right) \approx \exp \left( { - \Theta ^{\mathtt{\varsigma}} E_{\mathrm{B}}^{\mathtt{\varsigma}} D_{\max }^{{\mathrm{q,}\varsigma} }} \right) = {\varepsilon ^{{\mathrm{q,}}\mathtt{\varsigma} }},
\end{equation}
where ${\varepsilon ^{{\mathrm{q,}}\mathtt{\varsigma} }}$ is the queueing delay violation probability. If the service rate ${R^{\mathtt{\varsigma}} }$ is not less than $E_{\mathrm{B}}^{\mathtt{\varsigma}} $, then $\left( {D_{\max }^{{\mathrm{q,}\varsigma} },{\varepsilon_{\max } ^{{\mathrm{q,}}\mathtt{\varsigma} }}} \right)$ can be satisfied \cite{She2021Proc, LiuTCOM2025}, i.e., $\left( {{D^{{\mathrm{q,}}{\mathtt{\varsigma}} }} \le D_{\max }^{{\mathrm{q,}}{\mathtt{\varsigma}} },{\varepsilon ^{{\mathrm{q,}}{\mathtt{\varsigma}} }} \le \varepsilon _{\max }^{{\mathrm{q,}}{\mathtt{\varsigma}} }} \right)$. \par

Hence, based on EB and \textit{Lemmas} \ref{Lemma_NA_max_evaluation} $\sim$ \ref{Lemma_PPSNR}, we derive the expression for the LB on NA of AHetNets in \textit{Proposition} \ref{Proposition_NA_LB_AHetNets}, and analyze the impact of heterogeneity on the LB in \textit{Result} \ref{Result_impact_heterogeneity}. Especially, the degree of heterogeneity $U$ is characterized by the number of delay constraints and mobility compositions. Under a composition $\left({D_{\max }^{ \mathtt{\varsigma} }}, {{v_{\mathrm{r}}}}\right)$, the number of coherence intervals is ${\NC}\left( {D_{\max }^{\mathtt{\varsigma}} ,{v_{\mathrm{r}}}} \right) = \left\lceil {{{D_{\max }^{\mathtt{\varsigma}} } \mathord{\left/
 {\vphantom {{D_{\max }^{\mathtt{\varsigma}} } {{T_{\mathrm{C}}}\left( {{v_{\mathrm{r}}}} \right)}}} \right.
 \kern-\nulldelimiterspace} {{T_{\mathrm{C}}}\left( {{v_{\mathrm{r}}}} \right)}}} \right\rceil $. As mentioned in Section \ref{Sec_Packet_Delivery_Mechanism}, ${T_{\mathrm{S}}}$ is the greatest common divisor across diverse $D_{\max }^{ \mathtt{\varsigma} }$ and ${{T_{\mathrm{C}}}\left( {{v_{\mathrm{r}}}} \right)}$. Thus, ${T_{\mathrm{S}}}$ and $\Nf$ are functions of $U$, where ${T_{\mathrm{S}}} \left( U \right)$ and $\Nf \left( U \right)$ will reduce as $U$ increases due to more diverse delay constraints and UE mobility.
\begin{figure*}[t]
\centering
\setlength{\abovecaptionskip}{-0.15cm}
\begin{equation}\label{eq_NA_all_REC}
\begin{aligned} 
   \eta _{{\mathrm{LB}}}^{\mathrm{H}}\left( {U,K,\xi } \right)  
   = & {\mathbb{P}_\psi }\left( {\hat \gamma _{i,n,k}^{{\mathrm{ZF,LB}}} \left( {K,\xi } \right) \ge \max \left\{ {\gamma _{{\mathrm{th}}}^{\mathtt{\varsigma}} \left( {{{\left( {D_{\max }^{\mathrm{q},\mathtt{\varsigma} }} \right)}_{{u_{\mathrm{I}}}}}},U,K,\xi \right),\forall {u_{\mathrm{I}}} \in \left[ {1,{U_{\mathrm{I}}}} \right]} \right\}} \right) \times  \\
  &  \prod\nolimits_{{u_{{\mathrm{II}}}} = 1}^{{U_{{\mathrm{II}}}}} {\prod\nolimits_{j = 1}^{{\NC}\left( {{{\left( {D_{\max }^{ \mathtt{\varsigma} },{v_{\mathrm{r}}}} \right)}_{{u_{{\mathrm{II}}}}}}} \right)} {{\mathbb{P}_\psi }\left\{ {\hat \gamma _{i + \left( {j - 1} \right)\Ns,n,k}^{{\mathrm{ZF,LB}}} \left( {K,\xi } \right) \ge \gamma _{{\mathrm{th}}}^{\mathtt{\varsigma}} \left( {{{\left( {D_{\max }^{ \mathrm{q},\mathtt{\varsigma} }} \right)}_{{u_{{\mathrm{II}}}}}}}, U,K,\xi \right)} \right\}} }, \mathtt{\varsigma}  \in \left\{ \LSDC, \MSDC \right\}.
\end{aligned}
\end{equation} 
\hrulefill
\vspace{-1.5em}
\end{figure*}
\vspace{-0.1cm}
\begin{Proposition}\label{Proposition_NA_LB_AHetNets}
The LB on NA of FAP CoMP-enabled AHetNets with heterogeneity that stems from diverse delay constraints and UE mobility is expressed in \eqref{eq_NA_all_REC} at the top of the page. Achieving $\left( \eta_{\mathrm{LB}}^{\mathrm{H}} \ge \eta_{\max}, \eta_{\max} \to 1 \right)$ implies that the target NA (i.e., $\eta_{\max} \to 1$) is achieved. In \eqref{eq_NA_all_REC}, ${\left( {D_{\max }^{\mathrm{q},\mathtt{\varsigma} }} \right)_{{u_{\mathrm{I}}}}}$, $\forall {u_{\mathrm{I}}} \in \left[ {1,{U_{\mathrm{I}}}} \right]$ and ${\left( {D_{\max }^{\mathrm{q},\mathtt{\varsigma} }} \right)_{{u_{{\mathrm{II}}}}}}$, $\forall {u_{\mathrm{II}}} \in \left[ {1,{U_{\mathrm{II}}}} \right]$ denote the queueing delay constraints when $D_{\max }^{ \mathtt{\varsigma} } \le {T_{\mathrm{C}}}$ and $D_{\max }^{ \mathtt{\varsigma} } > {T_{\mathrm{C}}}$, respectively; ${\left( {D_{\max }^{\mathtt{\varsigma} }, v_{\mathrm{r}}} \right)_{{u_{{\mathrm{II}}}}}}$ denote the delay constraints and UE mobility compositions under $D_{\max }^{ \mathtt{\varsigma} } \!>\! {T_{\mathrm{C}}}$. Moreover, $U \!=\! {U_{\mathrm{I}}} \!+\! {U_{{\mathrm{II}}}}$, where ${U_{\mathrm{I}}}$ and ${U_{\mathrm{II}}}$ are the numbers of delay constraints and UE mobility compositions under $D_{\max }^{ \mathtt{\varsigma} } \le {T_{\mathrm{C}}}$ and $D_{\max }^{ \mathtt{\varsigma} } \!>\! {T_{\mathrm{C}}}$, respectively. Finally, $\gamma _{{\mathrm{th}}}^{\mathtt{\varsigma}} \left( {D_{\max }^{{\mathrm{q},\mathtt{\varsigma}} },U,K,\xi} \right)$ are required thresholds of post-processing SNR for satisfying $\left( {D_{\max }^{\mathtt{\varsigma}} ,\varepsilon _{\max }^{\mathtt{\varsigma}} } \right)$, which are expressed as
\begin{align}
   & \gamma _{{\mathrm{th}}}^{\LSDC} \big( {\cal P}^{\LSDC} \big) \!
     = \! \phi\exp \Big( {\frac{{{\varpi ^{\LSDC}}E_{\mathrm{B}}^{\LSDC}\left(D_{\max }^{{\mathrm{q},\LSDC} } \right)\ln 2}}{{B\left( K \right)\Nf\left( U \right){T_{\mathrm{f}}} - \xi}}} \Big), \label{apx_Proposition_NA_LDC_1} \\
   &  \gamma _{{\mathrm{th}}}^{\MSDC} \big( {\cal P}^{\MSDC} \big) \!  =  \notag \\
    & \exp  \Big( {\frac{{{\varpi ^{\MSDC}}E_{\mathrm{B}}^{\MSDC}\left(D_{\max }^{{\mathrm{q},\MSDC} } \right)\ln 2}}{{B\left( K \right)\Nf\left( U \right){T_{\mathrm{f}}} - \xi}} \!+\! \frac{{{f_{{Q^{ - 1}}}}\left( {\varepsilon _{\max }^{{\mathrm{u,c}}}} \right)}}{{\sqrt {B( K)T_{\mathrm{f}}^{\left( {\mathrm{d}} \right)}} }}} \Big) \! \!-\! 1, \label{apx_Proposition_NA_MDC_3}
\end{align}
where $\phi > 1$ is the SNR loss due to the finite blocklength \cite{Dong2021TWC}; ${\cal P}^{\mathtt{\varsigma}} \!=\! \left\{ {D_{\max }^{{\mathrm{q},\mathtt{\varsigma}} },U,K,\xi} \right\}$, $\mathtt{\varsigma}  \in \left\{ \LSDC, \MSDC \right\}$; $\varepsilon _{\max }^{{\mathrm{u,c}}}  \!=\! \varepsilon _{\max }^{{\MSDC}} \!- \varepsilon _{\max }^{{\mathrm{q,}\MSDC}}$ is the requirement of decoding error probability.
\end{Proposition}
\vspace{-0.3cm}
\begin{proof}
    See Appendix \ref{appendix_Proposition_NA_LB_AHetNets}.
\end{proof}
\vspace{-0.1cm}

\vspace{-0.1cm}
\begin{Result}\label{Result_impact_heterogeneity}
From \eqref{eq_NA_all_REC}, $\eta _{{\mathrm{LB}}}^{\mathrm{H}}$ decreases under extended heterogeneity, even when the heterogeneity stems from additional REC services under LSDC or from UEs with lower mobility. The reason is that $\gamma _{{\mathrm{th}}}^{\mathtt{\varsigma}} \left( {\cal P}^{\mathtt{\varsigma}} \right)$ in \eqref{apx_Proposition_NA_LDC_1} and \eqref{apx_Proposition_NA_MDC_3} increase with extended heterogeneity, which is due to the reduction in the time-frequency resources $B\left( K \right)\Nf \left( U \right){T_{\mathrm{f}}}$ for each service with increased $U$. 
\end{Result}
\vspace{-0.1cm}


\vspace{-0.1cm}
\subsection{NA Enhancement}
Based on the insights in \textit{Result} \ref{Result_impact_heterogeneity}, it is necessary to improve resource efficiency to enhance the LB on NA for mitigating the impact of heterogeneity. Thus, we analytically obtain \textit{Property} \ref{Property_joint_K_LoP} and formulate the optimization problem in this part.

\begin{Property}\label{Property_joint_K_LoP}
Since $B \left( K \right) \!\!=\!\! \left\lfloor {{{{B^{{\mathrm{tot}}}}} \mathord{\left/
 {\vphantom {{{B^{{\mathrm{tot}}}}} {\left( {{B_0}\left\lceil {{M \mathord{\left/
 {\vphantom {M K}} \right.
 \kern-\nulldelimiterspace} K}} \right\rceil } \right)}}} \right.
 \kern-\nulldelimiterspace} {\left( {{B_0}\left\lceil {{M \mathord{\left/
 {\vphantom {M K}} \right.
 \kern-\nulldelimiterspace} K}} \right\rceil } \right)}}} \right\rfloor {B_0}$, the time-frequency resources of each service can be increased under a larger $K$. However, $\hat \gamma _{i,n,k}^{{\mathrm{ZF,LB}}} \! \left( {K,\xi } \right)$ in \eqref{eq_PPSNR_UL_LB} decreases with larger $K$ due to increased interference and reduced spatial diversity, degrading $\eta _{{\mathrm{LB}}}^{\mathrm{H}}$. Moreover, a higher $\xi$ may be required due to $\xi \!\ge\! K$ in \eqref{eq_constraint_LoP}, which shortens the time-frequency resources for data transmission, i.e., ${{B\left( K \right)\Nf\left( U \right){T_{\mathrm{f}}} - \xi}}$. Consequently, $\gamma _{{\mathrm{th}}}^{\mathtt{\varsigma}} \left( {\cal P}^{\mathtt{\varsigma}} \right)$ in \eqref{apx_Proposition_NA_LDC_1} and \eqref{apx_Proposition_NA_MDC_3} increase, further degrading $\eta _{{\mathrm{LB}}}^{\mathrm{H}}$. Therefore, $K$ and $\xi$ should be jointly optimized to improve spatial, frequency, and temporal resource efficiency.
\end{Property}
\vspace{-0.1cm}

Based on \textit{Property} \ref{Property_joint_K_LoP}, we formulate the following optimization problem to jointly optimize $K$ and $\xi $. \begin{align}\label{eq_problem_joint_K_LoP}
{\mathop {\max }\limits_{K,\xi } } \quad & \eta _{{\mathrm{LB}}}^{\mathrm{H}}\left( {U,K,\xi } \right) \\
{{\mathrm{s.t.}}}  \quad & \eqref{eq_constraint_N}, \eqref{eq_constraint_K}, \eqref{eq_constraint_LoP}. \notag
\end{align}

The optimal values of $K$ and $\xi$, denoted by $K^\star$ and $\xi^\star$, respectively, for the problem \eqref{eq_problem_joint_K_LoP}, can be obtained using the exhaustive search method. The complexity of this method is ${\cal O} \left( {\left( {{\xi _{\max }} \!+ 1} \right)\left( {{K_{{\mathrm{II}}}} -\! {K_{\mathrm{I}}} +\! 1} \right) \!-0.5 {{\left( {{K_{{\mathrm{II}}}} \!+\! {K_{\mathrm{I}}}} \right)\left( {{K_{{\mathrm{II}}}} \!-\! {K_{\mathrm{I}}} +\! 1} \right)}}} \right)$.

\section{Simulation and Numerical Results}\label{Sec_sim_num}
In this section, we validate our analysis and results by presenting numerical simulations through $10^{10}$ Monte-Carlo trials. For our simulations, we consider a disaster area with a radius of $W_{\mathrm{D}}$ = 3,000 m, where $M$ = 1,000 UEs are all activated. The FAP flight altitude is set as $h_{\mathrm{F}}$ = 200 m. The AUE flight altitude range is set as ${\cal H} = \left[ {100,400} \right]$ m. The $\kappa$-$\mu$ shadowed fading parameters are set as $\kappa \to 0$, $\mu= 3$, and $m \to \infty$ for characterizing the Nakagami-$m$ fading, which is generally used in aerial networks \cite{LiuTCOM2025}. The SNR losses due to wireless fronthaul and finite blocklength are set as $\vartheta = 0.5$ and $\phi = 1.5$, respectively. The requirement of uOPL probability is set as $\varepsilon _{\max }^{  \mathtt{\varsigma} }$ = $10^{-5}$, $\mathtt{\varsigma}  \in \left\{ \LSDC, \MSDC \right\}$, where ${\varepsilon _{\max }^{{\mathrm{u,c}}}}$ = $\varepsilon _{\max }^{\mathrm{q,}\MSDC}$ = $\frac{1}{2}\varepsilon _{\max }^{\MSDC}$ and $\varepsilon _{\max }^{\mathrm{q},\LSDC}$ = $\varepsilon _{\max }^{\LSDC}$. In disaster emergency scenarios, the total delay bound in uplink RECs under MSDC is $D_{\max }^{{\MSDC}}$ = 0.5 ms, and the relative velocity is ${v_{\mathrm{r}}}$ = 30 (m/s) \cite{LiuTCOM2025}. For uplink RECs under LSDC, there are four groups of total delay bound $D_{\max }^{{\LSDC}}$ and relative velocity ${v_{\mathrm{r}}}$, which act as four sources of heterogeneity. The details of the four sources are shown in Table \ref{four_sources}. The degree of heterogeneity $U$ = $Z$, $Z \ge 2$ indicates that the scenario contains RECs under MSDC and LSDC with $\left( {D_{\max }^{{\LSDC}}\left( 1 \right),{v_{\mathrm{r}}}\left( 1 \right)} \right) \sim \left( {D_{\max }^{{\LSDC}}\left( {Z - 1} \right),{v_{\mathrm{r}}}\left( {Z - 1} \right)} \right)$; $U$ = 1 refers to the scenario that contains RECs under MSDC. The target NA is set as ${\eta _{\max }} = 0.98$. The total bandwidth is ${{B^{{\mathrm{tot}}}}}$ = 150 MHz. We consider the service stochastic arrival process as the Poisson process, where $\theta^{\mathtt{\varsigma}}$ = $\sigma _\mathtt{\varsigma} ^2$, ${\lambda ^{\LSDC}}$ = 500 packets/s, and ${\lambda ^{\MSDC}}$ = 20 packets/s. The EB (in packets/slot) for a Poisson arrival process is $E_{\mathrm{B}}^{\mathtt{\varsigma}}  = \frac{{{\Nf}{T_{\mathrm{f}}}\ln \left( {{1 \mathord{\left/
 {\vphantom {1 {{\varepsilon_{\max} ^{{\mathrm{q,}}\mathtt{\varsigma} }}}}} \right.
 \kern-\nulldelimiterspace} {{\varepsilon_{\max} ^{{\mathrm{q,}}\mathtt{\varsigma} }}}}} \right)}}{{D_{\max }^{{\mathrm{q,}}\mathtt{\varsigma} }\ln \left[ {{{{T_{\mathrm{f}}}\ln \left( {{1 \mathord{\left/
 {\vphantom {1 {{\varepsilon_{\max} ^{{\mathrm{q,}}\mathtt{\varsigma} }}}}} \right.
 \kern-\nulldelimiterspace} {{\varepsilon_{\max} ^{{\mathrm{q,}}\mathtt{\varsigma} }}}}} \right)} \mathord{\left/
 {\vphantom {{{T_{\mathrm{f}}}\ln \left( {{1 \mathord{\left/
 {\vphantom {1 {{\varepsilon_{\max} ^{{\mathrm{q,}}\mathtt{\varsigma} }}}}} \right.
 \kern-\nulldelimiterspace} {{\varepsilon_{\max} ^{{\mathrm{q,}}\mathtt{\varsigma} }}}}} \right)} {\left( {\theta^{\mathtt{\varsigma}}D_{\max }^{{\mathrm{q,}}\mathtt{\varsigma} }} \right)}}} \right.
 \kern-\nulldelimiterspace} {\left( {\theta^{\mathtt{\varsigma}}D_{\max }^{{\mathrm{q,}}\mathtt{\varsigma} }} \right)}} + 1} \right]}}$ \cite{LiuTCOM2025}. The other parameters are listed in Table \ref{simulation_settings}.

\begin{table}[tbp]
\vspace{-0.2cm}
\setlength{\abovecaptionskip}{0cm}
\caption{Sources of Heterogeneity}\label{four_sources}
\centering
\small
\begin{tabular}{|m{18 mm}<\centering|m{10 mm}<\centering|m{11 mm}<\centering|m{11 mm}<\centering|m{10 mm}<\centering|}
\hline
Sources & S$_1$ &  S$_2$ & S$_3$ & S$_4$ \\
\hline
$D_{\max }^{{\LSDC}}$ (ms) &  55:5:105 &  55  & 55:5:105  & 55:-5:5  \\
\hline
${v_{\mathrm{r}}}$ (m/s) & 30 & 30:-2:10 & 30:-2:10 & 30:2:50  \\
\hline
\end{tabular}
\vspace{-0.3cm}
\end{table}

\begin{table}[t]
\vspace{-0.2cm}
\setlength{\abovecaptionskip}{0cm}
\caption{Simulation Parameters \cite{LiuTCOM2025,Liu2025iotj, Elhoushy2021TWC}}\label{simulation_settings}
\centering
\small
\begin{tabular}{|m{56 mm}<\centering|m{18 mm}<\centering|}
\hline
\bfseries Parameter & \bfseries Value \\
\hline
\hline
Duration of each frame ${T_{\mathrm{f}}}$ (equals to TTI) & 0.1 ms \\
\hline
Orthogonal subcarrier spacing $B_{0}$ & 15 KHz \\
\hline
Carrier frequency $f_{\mathrm{c}}$  & 2 GHz \\
\hline
Sampling time ${{t_{\mathrm{s}}}}$  & 66.66 ${\mathrm{\mu s}}$ \\
\hline
Channel coherence bandwidth ${B_{\mathrm{C}}}$ & 0.5 MHz \\
\hline
Noise power spectral density ${{\mathcal N}}$ & -174 dBm/Hz \\
\hline
Transmit power ${\rho _{\mathrm{t}}} $ & 5 dBm \\
\hline
Packet size ${\varpi ^{\LSDC}}$/${\varpi ^{\MSDC}}$ & 1000/160 bits \\
\hline
\end{tabular}
\vspace{0.2em}
\end{table}

Figure \ref{NA_vs_U} shows the impact of different sources of heterogeneity on $\eta_{\mathrm{LB}}^{\mathrm{H}}$. The results indicate that as the degree of heterogeneity $U$ increases, $\eta_{\mathrm{LB}}^{\mathrm{H}}$ experiences significant degradation, even when the heterogeneity stems from additional REC services under LSDC or from UEs with lower mobility. Notably, when heterogeneity is driven by both additional REC services under LSDC and UEs with higher mobility, as in the source of S$_4$, the degradation in $\eta_{\mathrm{LB}}^{\mathrm{H}}$ becomes super-linear. These findings validate \textit{Result} \ref{Result_impact_heterogeneity}, confirming the pronounced sensitivity of $\eta_{\mathrm{LB}}^{\mathrm{H}}$ to extended heterogeneity in AHetNets.\par

\begin{figure}[t]
\vspace{-0.3cm}
\centering
\setlength{\abovecaptionskip}{-0.15cm}
\includegraphics[width=3.6in]{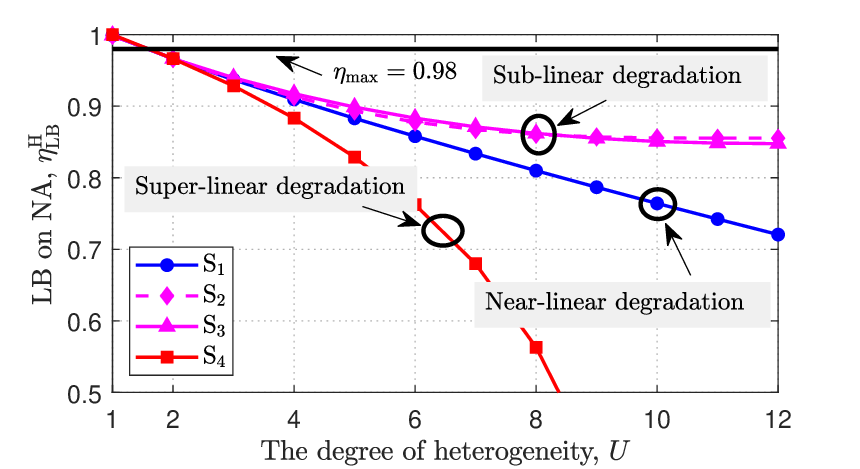}
\caption{LB on NA vs. $U$ with $L$ = 15, $K$ = 6, and $T_{\mathrm{f}}^{\left( {\mathrm{c}} \right)}$ = 0.05 ms.} 
\label{NA_vs_U}
\vspace{-0.8em}
\end{figure}

Figure \ref{NA_vs_U_diff_K_LoP} shows the impact of heterogeneity—caused by S$_4$ on $\eta _{{\mathrm{LB}}}^{\mathrm{H}}$ under various resource allocation policies. These policies include: existing \textit{(i)} fixed values for both $K$ and $\xi$, \textit{(ii)} optimization of $K$ only, \textit{(iii)} optimization of $\xi$ only, and our proposed \textit{(iv)} joint optimization of $K$ and $\xi$. In the figure, the legend ``{Baseline}" corresponds to the policy that $K = 6$ and $\xi$ is fixed, specifically with $T_{\mathrm{f}}^{(\mathrm{c})} = 0.05$ ms.  
The legend ``Opt. $K$" represents the policy that only $K$ is optimized, while $\xi$ remains fixed. Similarly, ``Opt. $\xi$" refers to optimizing $\xi$ with $K$ fixed at $6$. The legend ``Joint Opt. $K$ and $\xi$" indicates the policy that both $K$ and $\xi$ are jointly optimized. Besides, ``Max $U$" is the maximum $U$, where $\left( \eta_{\mathrm{LB}}^{\mathrm{H}} \ge \eta_{\max}, \eta_{\max} \to 1 \right)$ can be achieved. Compared to existing policies, the joint optimization of $K$ and $\xi$ significantly enhances the ability to achieve $\left( \eta_{\mathrm{LB}}^{\mathrm{H}} \ge \eta_{\max}, \eta_{\max} \to 1 \right)$ in AHetNets, especially under increased heterogeneity. This demonstrates that the proposed joint optimization policy effectively mitigates NA degradation caused by extended heterogeneity, thereby validating the effectiveness of \textit{Property} \ref{Property_joint_K_LoP} and the optimization problem in \eqref{eq_problem_joint_K_LoP} for achieving the target NA in AHetNets.

\begin{figure}[t]
\vspace{-0.2cm}
\centering
\setlength{\abovecaptionskip}{-0.15cm}
\includegraphics[width=3.6in]{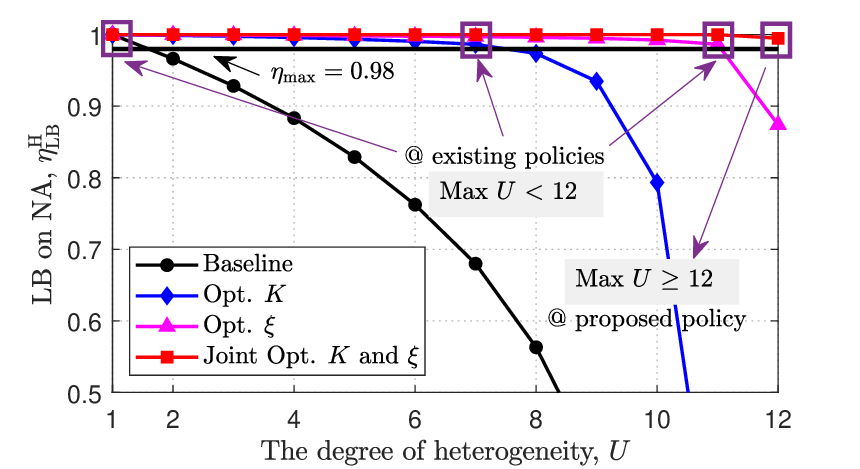}
\caption{LB on NA vs. $U$ under different $K$ and $\xi $ with $L$ = 15.} 
\label{NA_vs_U_diff_K_LoP}
\vspace{0.3em}
\end{figure}

\vspace{-0.3em}
\section{Conclusion}\label{Sec_Conclusion}
We investigated the NA of FAP CoMP-enabled AHetNets for delay-constrained REC services to UEs with varying mobility. To address resource allocation conflicts from heterogeneity in delay constraints and mobility, we employed a URA scheme. We derived expressions for the LB on NA, revealing that extended heterogeneity significantly degrades it. To mitigate this, we jointly optimized the number of UEs sharing time-frequency resources ($K$) and the pilot length ($\xi$), improving the LB. Simulations confirmed that target NA levels are achievable even under greater heterogeneity through joint optimization of $K$ and $\xi$.

\vspace{-0.3em}
\appendices
\section{Proof of Proposition \ref{Proposition_NA_LB_AHetNets}}
\label{appendix_Proposition_NA_LB_AHetNets}
\vspace{-0.6em}
Under diverse $D_{\max }^{ \mathtt{\varsigma} }$ and ${T_{\mathrm{C}}}$, it is likely to observe $D_{\max }^{ \mathtt{\varsigma} } \le {T_{\mathrm{C}}}$ and $D_{\max }^{ \mathtt{\varsigma} } > {T_{\mathrm{C}}}$, $\mathtt{\varsigma}  \in \left\{ \LSDC, \MSDC \right\}$, which result in quasi-static and time-varying channels, respectively \cite{LiuTCOM2025,Elhoushy2021TWC}.  \par

\subsubsection{NA when $D_{\max }^{{\LSDC}} > {T_{\mathrm{C}}}$}
For delay-sensitive services, even if the delay constraints are less stringent, the blocklength of channel coding is finite \cite{Dong2021TWC}. Then, with a specific SNR $\gamma$, the achievable service rate (packets/slot) under LSDC is \cite{Dong2021TWC}
\begin{equation}\label{eq_Rate_LDC}
\setlength{\abovedisplayskip}{0pt}
\setlength{\belowdisplayskip}{0pt}
{R^{{\LSDC}}} = \frac{{B\Nf T_{\mathrm{f}}^{\left( {\mathrm{d}} \right)}}}{{{\varpi ^{{\LSDC}}}\ln 2}}C\left( {\frac{\gamma }{\phi }} \right),
\end{equation}
where $C\left( \gamma  \right) = \ln \left( {1 + \gamma } \right)$ is Shannon’s capacity in the infinite blocklength regime; $BT_{\mathrm{f}}^{\left( {\mathrm{d}} \right)}$ is the blocklength of channel coding during each frame (in the number of channel uses); $B\Nf T_{\mathrm{f}}^{\left( {\mathrm{d}} \right)} = B\Nf{T_{\mathrm{f}}} - \xi $ is the time-frequency resources for data transmission of each service during each time slot.  
\par

From \cite{Dong2021TWC}, packet loss comes from queueing delay violations for the services under LSDC. Then, the QoS requirement $\left( {D_{\max }^{{\LSDC}},\varepsilon _{\max }^{{\LSDC}}} \right)$ for ensuring REC services under LSDC can be equated to $\left( {D_{\max }^{{\mathrm{q},\LSDC}},\varepsilon _{\max }^{{\mathrm{q},\LSDC}}} \right)$. In time-varying channels, $\hat \gamma _{i,n,k}^{{\mathrm{ZF,LB}}} \!\! = \! \hat \gamma _{i',n,k}^{{\mathrm{ZF,LB}}}, \forall i,i' \in \left[ {0,I} \right],\left| {i - i'} \right| \le \Ns$ and $\hat \gamma _{i,n,k}^{{\mathrm{ZF,LB}}} \!\!=\!\! \hat \gamma _{i + \Ns,n,k}^{{\mathrm{ZF,LB}}} \!\cup\! \hat \gamma _{i,n,k}^{{\mathrm{ZF,LB}}} \!\!\ne\!\! \hat \gamma _{i + \Ns,n,k}^{{\mathrm{ZF,LB}}}$. Consequently, ${R^{\LSDC}}$ varies across different ${T_{\mathrm{C}}}$ within $D_{\max }^{{\LSDC}}$ and remains constant during each ${T_{\mathrm{C}}}$. Then, $\left( {D_{\max }^{{\mathrm{q},\LSDC}},\varepsilon _{\max }^{{\mathrm{q},\LSDC}}} \right)$ can be satisfied if ${R^{\LSDC}} \ge E_{\mathrm{B}}^{\LSDC}$ can be assured during each ${T_{\mathrm{C}}}$ within $D_{\max }^{{\LSDC}}$. Thus, by substituting ${R^{\LSDC}} = E_{\mathrm{B}}^{\LSDC}$ into \eqref{eq_Rate_LDC}, we get the required threshold of post-processing SNR $\gamma _{{\mathrm{th}}}^{\LSDC}$ for satisfying $( {D_{\max }^{{\LSDC}},\varepsilon _{\max }^{{\LSDC}}} )$, which is given in \eqref{apx_Proposition_NA_LDC_1}. Based on \textit{Lemma} \ref{Lemma_NA_max_evaluation}, the LB on NA for REC services under LSDC with identical $\left( {D_{\max }^{\LSDC} ,\varepsilon _{\max }^{\LSDC} } \right)$ of UEs with same mobility when $D_{\max }^{{\LSDC}} > {T_{\mathrm{C}}}$ can be given by
\begin{equation}\label{eq_LB_NA_LDC}
\setlength{\abovedisplayskip}{0pt}
\setlength{\belowdisplayskip}{2pt}
    \hat \eta _{{\mathrm{LB}}}^{\LSDC} = \prod\nolimits_{j = 1}^{{\NC}} {{\mathbb{P}_\psi }\big( {\hat \gamma _{i + \left( {j - 1} \right)\Ns,n,k}^{{\mathrm{ZF,LB}}} \ge \gamma _{{\mathrm{th}}}^{\LSDC} } \big)}.
\end{equation}

\subsubsection{NA when $D_{\max }^{{\MSDC}} \le {T_{\mathrm{C}}}$}
When the delay constraints are more stringent, the blocklength of channel coding is significantly shorter \cite{Dong2021TWC}. As a result, decoding errors cannot be ignored under MSDC. According to finite blocklength information theory, the achievable service rate $R^{\MSDC}$ (packets/slot) of UE with a specific SNR $\gamma $ and a given decoding error probability ${{\varepsilon ^{{\mathrm{u,c}}}}}$ can be given by \cite{Dong2021TWC} 
\begin{equation}\label{eq_rate_MDC}
\setlength{\abovedisplayskip}{0pt}
\setlength{\belowdisplayskip}{0pt}
    {R^{\MSDC}} \!= \!\frac{{B\Nf{T_{\mathrm{f}}} - \xi}}{{{\varpi ^{\MSDC}}\ln 2}}\left[ {C\left( \gamma  \right) \!-\! \sqrt {\frac{{V\left( \gamma  \right)}}{{BT_{\mathrm{f}}^{\left( {\mathrm{d}} \right)}}}} {f_{{Q^{\! - 1}}}}\left( {{\varepsilon ^{{\mathrm{u}},{\mathrm{c}}}}} \right)} \right],
\end{equation}
where $V\left( \gamma  \right) \!=\! 1 \!-\! \frac{1}{{{{\left( {1 + \gamma } \right)}^2}}} \le 1$ is the channel dispersion.

From \cite{Dong2021TWC,LiuTCOM2025}, packet loss comes from decoding errors and queueing delay violations for services under MSDC. According to \cite{LiuTCOM2025}, by integrating the decoding error probability ${\varepsilon ^{{\mathrm{u,c}}}}$ and the queueing delay violation probability ${\varepsilon ^{\mathrm{q,}\MSDC}}$, the uOPL probability is given by ${\varepsilon_{\scriptscriptstyle \psi } ^{{\MSDC}}} \le {\varepsilon ^{{\mathrm{u,c}}}} + {\varepsilon ^{\mathrm{q,}\MSDC}}$. For QoS requirements $\left( {D_{\max }^{{\MSDC}},\varepsilon _{\max }^{{\MSDC}}} \right)$, $\varepsilon _{\max }^{{\MSDC}}$ can be divided into $\varepsilon _{\max }^{{\mathrm{u,c}}}$ and $\varepsilon _{\max }^{{\mathrm{q,}\MSDC}}$ under the delay constraints $D_{\max }^{{\MSDC}}$, i.e., $\varepsilon _{\max }^{{\MSDC}} = \varepsilon _{\max }^{{\mathrm{u,c}}} + \varepsilon _{\max }^{{\mathrm{q,}\MSDC}}$, where $\varepsilon _{\max }^{{\mathrm{u,c}}}$ is the requirement of ${\varepsilon ^{{\mathrm{u,c}}}}$. In quasi-static channels, ${R^{\MSDC}}$ is constant within $D_{\max }^{{\MSDC}}$ \cite{LiuTCOM2025}. By substituting ${R^{\MSDC}} = E_{\mathrm{B}}^{\MSDC}$, ${\varepsilon ^{{\mathrm{u,c}}}} = \varepsilon _{\max }^{{\mathrm{u,c}}}$, and $V\left( \gamma  \right) \!=\! 1$ into \eqref{eq_rate_MDC}, we can obtain the required threshold of post-processing SNR $\gamma _{{\mathrm{th}}}^{\MSDC}$ for satisfying $\left( {D_{\max }^{{\MSDC}},\varepsilon _{\max }^{{\MSDC}}} \right)$, which is given in \eqref{apx_Proposition_NA_MDC_3}. \par

Based on \textit{Lemma} \ref{Lemma_NA_max_evaluation}, the LB on NA for REC services under MSDC with identical $\left( {D_{\max }^{\MSDC} ,\varepsilon _{\max }^{\MSDC} } \right)$ of UEs with same mobility when $D_{\max }^{{\MSDC}} \le {T_{\mathrm{C}}}$ can be expressed as
\begin{equation}\label{eq_LB_NA_MDC}
\setlength{\belowdisplayskip}{0pt}
    {\tilde \eta } _{{\mathrm{LB}}}^{{\MSDC}} = {\mathbb{P}_\psi }\big( {\hat \gamma _{i,n,k}^{{\mathrm{ZF,LB}}} \ge \gamma _{{\mathrm{th}}}^{\MSDC} } \big).
\end{equation}

\subsubsection{NA when $D_{\max }^{{\LSDC}} \le {T_{\mathrm{C}}}$}
Based on \eqref{apx_Proposition_NA_LDC_1} and \eqref{eq_LB_NA_MDC}, the LB on NA for REC services under LSDC with identical $\left( {D_{\max }^{\LSDC} ,\varepsilon _{\max }^{\LSDC} } \right)$ of UEs with same mobility, when $D_{\max }^{{\LSDC}} \le {T_{\mathrm{C}}}$, can be given by
\begin{equation}\label{eq_LB_NA_MDC_2}
\setlength{\abovedisplayskip}{0pt}
\setlength{\belowdisplayskip}{0pt}
    {\tilde \eta } _{{\mathrm{LB}}}^{{\LSDC}} = {\mathbb{P}_\psi }\big( {\hat \gamma _{i,n,k}^{{\mathrm{ZF,LB}}} \ge \gamma _{{\mathrm{th}}}^{\LSDC} } \big).
\end{equation}

\subsubsection{NA under $D_{\max }^{{\MSDC}} > {T_{\mathrm{C}}}$}
Based on \eqref{apx_Proposition_NA_MDC_3} and \eqref{eq_LB_NA_LDC}, the LB on NA for REC services under MSDC with identical $\left( {D_{\max }^{\MSDC} ,\varepsilon _{\max }^{\MSDC} } \right)$ of UEs with same mobility, when $D_{\max }^{{\MSDC}} > {T_{\mathrm{C}}}$, can be given by
\begin{equation}\label{eq_LB_NA_LDC_2}
\setlength{\abovedisplayskip}{0pt}
\setlength{\belowdisplayskip}{0pt}
    {\hat \eta } _{{\mathrm{LB}}}^{{\MSDC}} = \prod\nolimits_{j = 1}^{{\NC}} {{\mathbb{P}_\psi }\big( {\hat \gamma _{i + \left( {j - 1} \right)\Ns,n,k}^{{\mathrm{ZF,LB}}} \ge \gamma _{{\mathrm{th}}}^{\MSDC} } \big)}.
\end{equation}

Since $\tilde{\eta}_{\mathrm{LB}}^{\mathtt{\varsigma}}$, $\mathtt{\varsigma} \!\!\in\!\! \{\LSDC,\MSDC\}$, in \eqref{eq_LB_NA_MDC} and \eqref{eq_LB_NA_MDC_2} depend only on $\gamma_{\mathrm{th}}^{\mathtt{\varsigma}}$, ensuring the REC with a higher $\gamma_{\mathrm{th}}^{\mathtt{\varsigma}}$ automatically ensures the REC with a lower $\gamma_{\mathrm{th}}^{\mathtt{\varsigma}}$. Ensuring REC when $D_{\max}^{\mathtt{\varsigma}} > T_{\mathrm{C}}$ is independent of the case when $D_{\max}^{\mathtt{\varsigma}} \le T_{\mathrm{C}}$, since $\hat{\eta}_{\mathrm{LB}}^{\mathtt{\varsigma}}$ in \eqref{eq_LB_NA_LDC} and \eqref{eq_LB_NA_LDC_2} depends on both $\gamma_{\mathrm{th}}^{\mathtt{\varsigma}}$ and $\NC$. Moreover, when $D_{\max}^{\mathtt{\varsigma}} > T_{\mathrm{C}}$, REC under different delay constraints shows weak coupling across UEs with varying mobility. Specifically, ensuring REC with larger $\gamma_{\mathrm{th}}^{\mathtt{\varsigma}}$ or $\NC$ does not guarantee it for smaller values, even though $\hat{\eta}_{\mathrm{LB}}^{\mathtt{\varsigma}}$ increases as $\gamma_{\mathrm{th}}^{\mathtt{\varsigma}}$ or $\NC$ decreases. This is because a smaller $\NC$, due to stringent $D_{\max}^{\mathtt{\varsigma}}$, is accompanied by a larger $\gamma_{\mathrm{th}}^{\mathtt{\varsigma}}$. Finally, by combining NAs across delay constraints and mobility levels, we derive a LB on overall NA valid for both $D_{\max}^{\mathtt{\varsigma}} \le T_{\mathrm{C}}$ and $D_{\max}^{\mathtt{\varsigma}} > T_{\mathrm{C}}$. Based on this, \eqref{eq_NA_all_REC} follows, completing the proof.

\par

\bibliographystyle{IEEEtran}
\bibliography{myref}

@ARTICLE{Dao2021Tuts,
  author={Dao, Nhu-Ngoc and others},
  journal={IEEE Commun. Surveys Tuts.}, 
  title={Survey on Aerial Radio Access Networks: Toward a Comprehensive 6{G} Access Infrastructure}, 
  year={2nd Quarter 2021},
  volume={23},
  number={2},
  pages={1193-1225}}

@ARTICLE{Yao2022MNET,
  author={Yao, Zhuohui and Cheng, Wenchi and Zhang, Wei and Zhang, Tao and Zhang, Hailin},
  journal={IEEE Netw.}, 
  title={The Rise of {UAV} Fleet Technologies for Emergency Wireless Communications in Harsh Environments}, 
  year={Jul. 2022},
  volume={36},
  number={4},
  pages={28-37}}

@ARTICLE{LiuTCOM2025,
  author={Liu, Junyu and others},
  journal={IEEE Trans. Commun.}, 
  title={Towards Reliable Communications with Delay Requirement in Aerial Disaster Emergency Networks via Coordinated Multi-Point}, 
  year={Oct. 2025},
  volume={73},
  number={10},
  pages={8781-8796}}

@ARTICLE{She2021Proc,
  author={She, Changyang and others},
  journal={Proc. IEEE}, 
  title={A Tutorial on Ultrareliable and Low-Latency Communications in 6{G}: Integrating Domain Knowledge Into Deep Learning}, 
  year={Mar. 2021},
  volume={109},
  number={3},
  pages={204-246}}

@ARTICLE{Elhoushy2021TWC,
  author={Elhoushy, Salah and Hamouda, Walaa},
  journal={IEEE Trans. Wireless Commun.}, 
  title={Limiting {D}oppler Shift Effect on Cell-Free Massive {MIMO} Systems: A Stochastic Geometry Approach}, 
  year={Sept. 2021},
  volume={20},
  number={9},
  pages={5656-5671}}

@ARTICLE{Liu2025TVT,
  author={Liu, Luyang and Wu, Shaochuan and Ma, Yongkui},
  journal={IEEE Trans. Veh. Technol.}, 
  title={Resource Allocation in Cell-Free Massive {MIMO} with the Coexistence of e{MBB} and {URLLC}}, 
  year={Oct. 2025},
  volume={74},
  number={10},
  pages={16039-16054}}

@ARTICLE{Imran2024Proc,
  author={Imran, Muhammad A. and others},
  journal={Proc. IEEE}, 
  title={Exploring the Boundaries of Connected Systems: Communications for Hard-to-Reach Areas and Extreme Conditions}, 
  year={Jul. 2024},
  volume={112},
  number={7},
  pages={912-945}}

@ARTICLE{Mohammadi2024Proc,
  author={Mohammadi, Mohammadali and Mobini, Zahra and Q. Ngo, Hien and Matthaiou, Michail},
  journal={Proc. IEEE}, 
  title={Next-Generation Multiple Access With Cell-Free Massive {MIMO}}, 
  year={Sept. 2024},
  volume={112},
  number={9},
  pages={1372-1420}}

@ARTICLE{Dong2021TWC,
  author={Dong, Rui and others},
  journal={IEEE Trans. Wireless Commun.}, 
  title={Deep Learning for Radio Resource Allocation With Diverse Quality-of-Service Requirements in 5{G}}, 
  year={Apr. 2021},
  volume={20},
  number={4},
  pages={2309-2324}}

@book{3gpp_22_261_TS,
 title = {Service Requirements for the 5G System},
 address = {Technical Specification (TS) 22.261, 3rd Generation Partnership Project},
 year = {Jun. 2022}
}

@ARTICLE{Liu2025iotj,
  author={Liu, Zhaoqing and Li, Kang and Wang, Yan and Zhu, Pengcheng},
  journal={IEEE Internet Things J.}, 
  title={Mixed Traffic Scheduling With Latency and Jitter Analysis in {URLLC} Industrial Automation}, 
  year={Jun. 2025},
  volume={12},
  number={12},
  pages={18820-18835}}

@ARTICLE{Elwekeil:TCOM:2023,
  author={Elwekeil, Mohamed and others},
  journal={IEEE Trans. Commun.},
  title={Power Control in Cell-Free Massive {MIMO} Networks for {UAV}s {URLLC} Under the Finite Blocklength Regime}, 
  year={2023},
  volume={71},
  number={2},
  pages={1126-1140},
  ISSN={1558-0857},
  month={Feb},}

@ARTICLE{Zhang2025TVT,
  author={Zhang, Xiaoqing and others},
  journal={IEEE Trans. Veh. Technol.}, 
  title={Multi-Agent Deep Reinforcement Learning-based Uplink Power Control in Cell-Free Massive {MIMO} with Mobile Users}, 
  year={2025 (Early Access)},
  volume={},
  number={},
  pages={1-16}}

@ARTICLE{Abbas2025IoTJ,
  author={Abbas, Mustafa S. and Mobini, Zahra and Ngo, Hien Quoc and Shin, Hyundong and Matthaiou, Michail},
  journal={IEEE Internet Things J.}, 
  title={Joint {AP} Selection and Power Allocation for Unicast-Multicast Cell-Free Massive {MIMO}}, 
  year={Nov. 2025},
  volume={12},
  number={22},
  pages={47135-47150}}

\end{document}